\documentclass[12pt]{article}
\usepackage[utf8]{inputenc}

\usepackage[textwidth=7.2in,top=1.2in,bottom=1.2in]{geometry}

\usepackage[margin=1cm]{caption}
\title{\vspace{-4cm} Learning Robust Treatment Rules for Censored Data}

\author{
Yifan Cui\thanks{Zhejiang University}~~~Junyi Liu\thanks{Tsinghua University}~~~Tao Shen\thanks{National University of Singapore}~~~Zhengling Qi\thanks{George Washington University}~~~Xi Chen\thanks{New York University}
}

\date{}

\usepackage{amsthm,amsmath,amssymb,bbm,bm}
\usepackage{natbib}
\usepackage{multirow}
\usepackage[pdftex]{graphicx}
\usepackage{wrapfig}
\usepackage{tikz}
\usepackage{centernot}
\usepackage{array}
\usepackage{url}
\usepackage{algorithmic}
\usepackage[linesnumbered, ruled, vlined]{algorithm2e}
\usepackage{mathrsfs}
\usepackage{dsfont}
\usepackage{titling}
\usepackage{relsize}
\usepackage{rotating}
\usepackage{enumitem}
\usepackage{titling}
\usepackage{float}
\usepackage{subcaption}
\usepackage{tabularx}
\usepackage{threeparttable}
\usepackage{booktabs}
\usepackage{verbatim}
\usepackage{ulem}

\newtheorem{remark}{Remark}

\newtheorem{theorem}{Theorem}[section]
\newtheorem{lemma}{Lemma}[section] 

{\textwidth=6.5in}

\usepackage{xr-hyper}
\usepackage{hyperref} 
\hypersetup{
  colorlinks = false,
  pdfborder  = {0 0 0},
}
\newcommand{\cX}{{\mathcal{X}}}

\newcommand{\sign}{{\text{sign}}}

\newcommand{\cL}{{\mathcal{L}}}
\newcommand{\cH}{{\mathcal{H}}}

\newcommand{\cF}{{\mathcal{F}}}

\newcommand{\Eb}{\mathbb{E}}

\newcommand{\Xb}{X}

\def\calH{{\cal H}}

\newcommand{\wh}{\widehat}
\newcommand{\wt}{\widetilde}

\newcommand{\st}[1]{\textcolor{black}{#1}\xspace}

\definecolor{DSgray}{cmyk}{0,1,0,0}

\begin{document}
\maketitle
\vspace{-2cm}
\abstract{
There is a fast-growing literature on estimating optimal treatment rules directly by maximizing the expected outcome.
In biomedical studies and operations applications, censored survival outcomes are frequently observed, in which case the truncated mean survival time and survival probability are of great interest.
In this paper, we propose two robust criteria for learning optimal treatment rules with censored survival outcomes; the former one targets an optimal treatment rule maximizing the truncated mean survival time, where the cutoff is specified by a given quantile such as the median; the latter one targets an optimal treatment rule maximizing buffered survival probabilities, where the predetermined threshold is adjusted to account for the truncated mean survival time.
We develop a sampling-based difference-of-convex algorithm for learning the proposed optimal treatment rules, and provide theoretical justifications for them.
In simulation studies, our estimators show improved performance compared to existing methods.  We also demonstrate the proposed method using AIDS clinical trial data.

}

\vspace{0.1cm}

\noindent {\bf keywords:}
Buffered probability of exceedance, Censoring, Conditional value-at-risk, Difference-of-convex algorithm, Robust treatment rules

\section{Introduction}
An individualized treatment rule provides a personalized treatment strategy for each patient in the population based on their individual characteristics. Recently, a significant amount of work has been devoted to estimating optimal treatment rules, e.g., see \cite{ qian2011performance,zhao2012estimating,zhang2012robust,chakraborty2013statistical,kitagawa2018should,kosorok2019precision,tsiatis2019dynamic,athey2021policy,Cui2021Individualized,stensrud2024optimal}.

\st{Mean-based criteria may be inadequate for skewed outcomes because gains among subjects with favorable outcomes can offset substantial losses among vulnerable subjects, causing a mean-optimal rule to perform poorly in the lower tail \citep{wang2017quantile,qi2019siam,leqi2021median,cui2023individualized}. This concern is particularly relevant for survival outcomes, where short survival times have serious clinical consequences and outcomes are often right-censored because of study termination or dropout. Similar censoring also arises in operational applications due to inventory or price constraints. These settings motivate interpretable, tail-sensitive treatment rules that accommodate censoring in a principled manner.}

In this paper, we propose two robust criteria to estimate optimal treatment rules for censored survival data in order to control the lower tail of the subjects' survival outcomes. The former targets an optimal treatment rule that maximizes the truncated mean survival time \citep{tian2014predicting,uno2014moving}, where the truncation is specified by a given quantile such as the median; the latter targets an optimal treatment rule maximizing buffered survival probabilities, where the predetermined threshold is adjusted to account for the truncated mean survival time. 
For the first criterion, we circumvent the need to specify an ad hoc cutoff but instead to specify a quantile level which is a more interpretable measure. Because the survival times are only partially observed, the conditional value-at-risk (CVaR) framework \citep{rockafellar2000optimization} and its application in individualized decision making \citep{qi2019siam,Qi2019EstimatingID} cannot be directly adopted. We propose an inverse-probability-weighting estimator by using the conditional survival function of the censoring distribution. 
For the second criterion, 
we focus on optimizing the probability of survival time exceeding a quality-adjusted survival time.
We then formally establish its connection with buffered probability-of-exceedance (bPOE) \citep{mafusalov2018buffered} and effectively utilize the proposed efficient optimization algorithm to estimate a robust treatment rule. 
Both CVaR and bPOE account for the tail behavior of the probability distribution so that our proposed criteria achieve robust maximization of the truncated mean survival time and the survival probability, respectively.
We refer to the two criteria as CVaR criterion and buffered criterion, respectively.  Recently, \cite{wang2017quantile,leqi2021median} studied different frameworks for estimating various quantile-optimal treatment rules. Our targeted criterion here is different from theirs as \cite{wang2017quantile} are interested in estimating the quantile-optimal treatment rule, i.e., maximizing the quantile of outcomes under a hypothetical intervention, and \cite{leqi2021median} define an optimal treatment rule by the quantile treatment effect.
In contrast, our CVaR and buffered criteria are particularly suitable for survival data analysis as they are designed to robustly maximize the truncated mean survival time as well as the survival probability.

 The treatment rule is an indicator function, so solving the corresponding optimization problems to the global optimality is NP-hard. We approximate the indicator function by a difference-of-convex (DC) function with a small accuracy threshold; such approximation is a customary approach that was utilized in \cite{zhou2017residual,Qi2019EstimatingID} for estimating individualized decision rules and \cite{hong2011sequential} for approximately solving chance-constrained stochastic programming. With the DC approximation strategy, the optimization problem for learning the treatment rule is of the DC kind as in \cite{Qi2019EstimatingID} for individualized decision rule, and thus a deterministic difference-of-convex algorithm (DCA) \citep{pang2016computing} is applicable to obtain a directional stationary solution of the corresponding optimization problem as in \cite{Qi2019EstimatingID}. However, the objective function for the treatment rule estimation is the minimizing value functions of the finite sum over all the data, which renders a summation over $O(n^2)$ terms in the DC decomposition and further results in the computational inefficiency of the deterministic DCA when the sample size is large. 

In order to improve the computational efficiency, we propose a new sampling-based DCA which solves a sequence of subproblems. We show that with the appropriate control of the incremental sampling rate, the limit point of the sequence of solutions obtained from the sampling-based DCA is almost surely a directional stationary point of the sharpest kind as in the deterministic DCA.  It is worth noting that the feature of an inner minimizing  function in the optimization problem distinguishes our sampling-based algorithm from the stochastic DC algorithm in \cite{an2019stochastic} for stochastic difference-of-convex programs and the sampling-based algorithm in \cite{liu2022risk} for stochastic difference-of-convex value-function optimization. As a result, the stochastic convexified subproblem constructed at each iteration is a biased approximation of the original objective function and this requires the special strategy utilizing the $\epsilon$-active index set in the update rule of the algorithm in order to obtain the sharpest-kind of stationary solutions; this will become clearer in Section~\ref{sec:algo} in the presentation of the optimization problem and the algorithm in detail.

\st{\subsection{Motivating examples and illustration}
Consider first biomedical applications in which the outcome is survival time or another clinically relevant measure. In oncology, a treatment rule may recommend radiotherapy according to patient characteristics. Maximizing overall survival can favor a treatment with large benefits for patients who already have relatively favorable outcomes, even when that treatment leads to substantially shorter survival. In this setting, it may instead be desirable to maximize the average survival time within the lower tail of the outcome distribution, corresponding to the criterion $V^1(d)$ introduced below. In a study involving a clinically meaningful biomarker endpoint, the objective may instead be to select a treatment or dosage that increases the probability that the endpoint exceeds a prespecified target, corresponding to $V^2(d)$. The first criterion therefore evaluates the magnitude of outcomes among worse-off subjects, whereas the second provides a probability-based assessment of achieving a favorable outcome.}

\st{Such criteria are also useful in supply-chain applications involving a supplier and a retailer. The retailer determines whether to increase or decrease its order quantity according to the wholesale price and other covariates. A conservative retailer may prefer a decision rule that improves profits in the lower tail rather than one that merely maximizes average profit, making $V^1(d_r)$ relevant. In a subsequent supplier-led decision, the supplier may instead seek to maximize the likelihood that profit exceeds a baseline target, motivating $V^2(d_s)$. Thus, the two criteria correspond to different operational objectives that cannot generally be recovered by maximizing the mean alone.}

\st{Figure~\ref{fig:1} gives an empirical illustration of this distinction. For the male group, the outcome under $A=1$ follows an accelerated failure time model with random error $N(1,1)$, whereas the outcome under $A=-1$ has random error $N(0.95,0.5)$. For the female group, these two distributions are reversed. We compare the mean criterion $V(d)$, the proposed CVaR criterion $V^1(d)$ with $\gamma=0.5$ and the buffered criterion $V^2(d)$ with $\tau=0.4$. The mean-optimal rule assigns $A=1$ to males and $A=-1$ to females because the corresponding distributions have slightly larger means. In contrast, the rules maximizing $V^1(d)$ and $V^2(d)$ reverse these assignments because the alternative distributions provide more favorable lower-tail behavior. As shown in Figure~\ref{fig:1}, although the mean-optimal rule has the largest value under $V(d)$, it performs considerably worse under $V^1(d)$ and $V^2(d)$. This example demonstrates that optimizing the mean does not generally protect against poor tail outcomes; see also Table~1 of \cite{wang2017quantile}.}

\begin{figure}[h]
\centering
\includegraphics[width=1\linewidth]{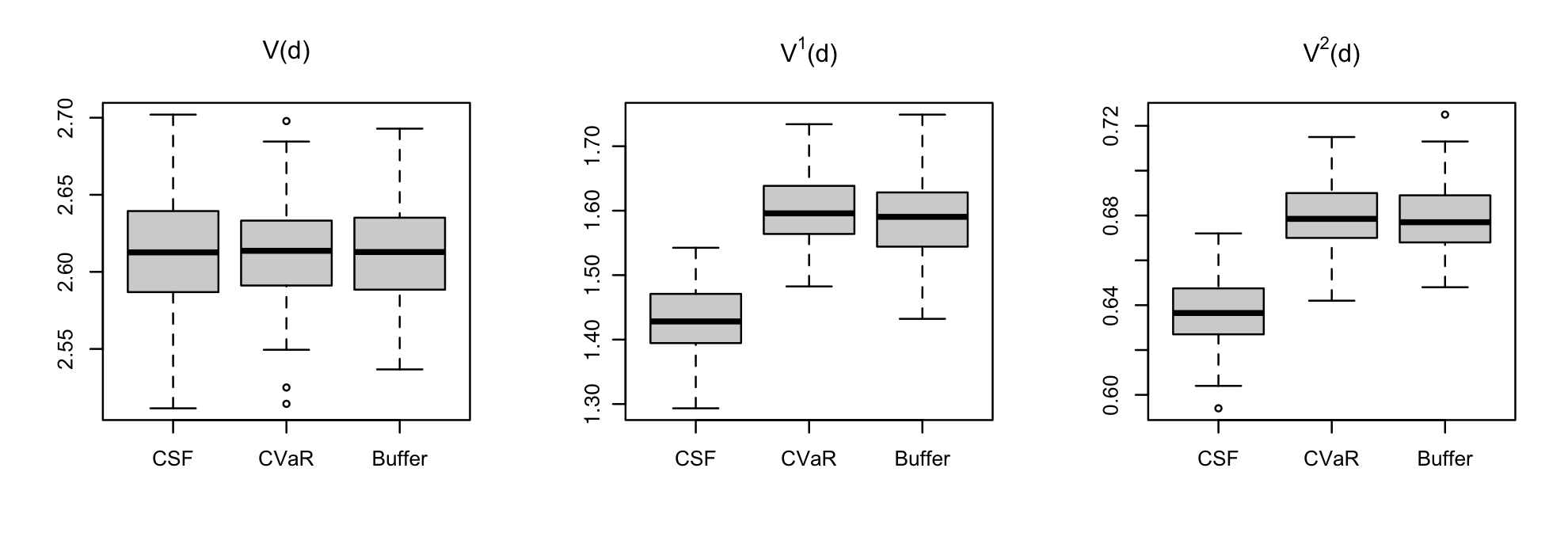}
\caption{An illustrative comparison of the mean, CVaR, and buffered criteria.}
\label{fig:1}
\end{figure}

\st{
\subsection{Our contributions}
The major contributions of this paper are threefold. First, we propose two robust criteria for learning optimal treatment rules.  The first CVaR criterion aims to maximize truncated mean survival time given a pre-specified quantile, where the quantile can be viewed as a quality-adjusted survival time determined by survival function, while the second buffered criterion maximizes survival function at a quality-adjusted survival time, where the quality-adjusted survival time is in turn determined by truncated mean survival time. In particular, we establish a formal link between optimizing truncated mean survival time and optimizing survival probabilities over a class of treatment rules in Lemma~\ref{lemma:maxmin} by effectively using the intrinsic connection between CVaR and bPOE. Second, we develop inverse probability weighted estimators for both CVaR and buffered criteria with censored survival data. Under certain regularity conditions, the estimated rules are shown to have nearly optimal performance in terms of value functions. Third, leveraging the optimization formulas of CVaR and bPOE and the DC approximation of the indicator function involved, we formulate robust treatment rule learning as a DC program. We then develop a sampling-based DCA that applies under both criteria and establish subsequential convergence to a directional stationary solution.
}
\st{
\subsection{Organization}
The remainder of the article is organized as follows. In Section~\ref{sec:method}, we lay out the setup for individualized treatment rules and introduce our robust criteria under right censoring.
In Section~\ref{sec:algo}, we develop a novel optimization algorithm to learn robust treatment rules. 
In Section~\ref{sec:theory}, we propose the theoretical justification of the proposed estimators. Extensive simulation studies are
presented in Section~\ref{sec:simu}.  Numerical studies demonstrate that the proposed method outperforms existing alternatives under the proposed criterion in a variety of settings. We also illustrate our method using AIDS clinical trials in Section~\ref{sec:real}. The article concludes with a discussion of future work in Section~\ref{sec:discussion}. Some needed technical results and additional discussions are provided in the supplementary material.
}

\section{Methodology}\label{sec:method}

Suppose $\{(X_i,A_i,T_i)\}_{i=1}^n$ are independent and identically distributed (i.i.d.) draws from an idealized uncensored data generating process,
where each $X_i\in\cX$ is a $p$-dimensional baseline covariate vector, $A_i\in\{1,-1\}$ is the treatment label, and $T_i$ denotes the outcome. Without loss of generality, we refer to $T_i$ as the survival time throughout for illustration.
Throughout, we assume standard consistency, positivity, and unconfoundedness assumptions in the causal inference literature.
To motivate our robust decision criteria,
Sections~\ref{subsec:0}--\ref{subsec:2} work in this setting to define the decision targets before introducing censoring in Section~\ref{subsec:rc}.

\subsection{Setup and original value function framework}\label{subsec:0}

Let $d:\mathcal X \to \{-1,1\}$ denote a treatment rule that assigns a treatment based on covariates $X \in \mathcal X$, and let $T(d)$ be the corresponding potential survival time under a hypothetical intervention that enforces rule $d$, which can be rewritten as
\[
T(d) \equiv T(d(X))
=
T(1) I\{d(X)=1\} + T(-1) I\{d(X)=-1\},
\]
where $T(a)$ is the potential survival time under an intervention that sets $A=a$.

In a randomized trial without censoring, the mean-optimal treatment rule is defined as the rule that maximizes the mean potential survival time $V(d) \equiv \mathbb{E}[T(d)]$.
It can be characterized pointwise as
\[
\bar d(X)
=
\operatorname{sign}\big\{\mathbb{E}[T \mid X,A=1] - \mathbb{E}[T \mid X,A=-1]\big\}.
\]
As shown by \citet{qian2011performance}, learning the mean-optimal individualized treatment rule is equivalently formulated as
\begin{equation}
\bar d
=
\arg\max_{d} V(d)
=
\arg\max_{d} \mathbb{E}[T(d)]
=
\arg\max_{d} \mathbb{E}\!\left[\frac{I\{d(X)=A\} T}{\pi(A \mid X)}\right],
\label{eq:opt1}
\end{equation}
where $\pi(A | X)$ is the propensity score \citep{10.2307/2335942}, which is known in randomized experiments.
The last equality shows that, under the idealized data generating process, $V(d)$ can be expressed solely in terms of $\{(X_i,A_i,T_i)\}_{i=1}^n$.

Rather than maximizing the value function in \eqref{eq:opt1} directly, \citet{zhao2012estimating,zhang2012estimating} recast the problem as the following equivalent weighted classification problem:
\begin{equation}
\bar d
=
\arg\min_{d} \mathbb{E}\!\left[\frac{T}{\pi(A \mid X)} I\{d(X)\neq A\}\right],
\label{eq:opt3}
\end{equation}
that is, a $0$–$1$ loss with weight $T/\pi(A|X)$.
A derivation of \eqref{eq:opt1} and \eqref{eq:opt3} is provided in the supplementary material for completeness.
To alleviate the computational burden of the $0$–$1$ loss in \eqref{eq:opt3}, \citet{zhao2012estimating} replace it with the hinge loss and solve the resulting problem via support vector machines.
The ensuing classification approach has appealing robustness properties, particularly in randomized studies where no outcome model is required.

\subsection{The CVaR criterion: maximizing truncated mean survival time at a given quantile}\label{subsec:1}

\st{In survival analysis, treatment decisions may seek to improve the lower tail of the survival distribution induced by an individualized treatment rule, rather than its overall mean. For a threshold $t_{0}>0$, we naturally consider the following truncated mean
\[
\Eb\bigl[T(d)\mid T(d)\le t_{0}\bigr].
\]
This criterion summarizes potential survival among individuals whose outcome under the rule does not exceed $t_{0}$, thereby emphasizing clinically unfavorable survival outcomes while excluding sufficiently long survival times from the evaluation.}

\st{The choice of $t_{0}$ is, however, not unique. It may be prespecified as a clinically meaningful survival horizon, although such a threshold can be ad hoc when no established clinical benchmark is available. Alternatively, if treatment $-1$ represents the standard of care, one may set $t_{0}=Q_{\gamma}\{T(-1)\}$,
where $Q_{\gamma}\{T(-1)\}$ denotes the right $\gamma$-quantile of the potential survival time under the reference treatment. This choice provides a common and estimable benchmark across treatment rules. However, although such choices permit direct comparison across rules, the proportion
$\Pr(T(d)\leq t_{0})$
generally varies with $d$. For some rules, this proportion may be large, so the resulting criterion may reflect the outcomes of a substantial fraction of patients rather than specifically characterizing the worst-performing tail. This may depart from our original goal of emphasizing the unfavorable tail of the survival distribution.}

To preserve a common lower-tail interpretation across treatment rules, we instead index the truncation threshold by a quantile of the potential survival time distribution induced by each rule.
For $\gamma\in(0,1)$, define
\[
Q_\gamma\{T(d)\}
= \inf\bigl\{\alpha\in\mathbb R:\Pr(T(d)\le \alpha) > \gamma\bigr\},
\]
the right $\gamma$-quantile of $T(d)$. Equivalently, $Q_\gamma\{T(d)\}$ is the time by which at least a proportion $\gamma$ of individuals would fail under rule $d$, so that the corresponding survival probability is at most $1-\gamma$. In many settings, specifying such a fraction is more straightforward than choosing a fixed time horizon $t_{0}$, which requires more detailed knowledge of the absolute survival scale.

\st{We then define our quantile-indexed truncated mean survival time as
\begin{equation}\label{eq:VVV}
V_\gamma^1(d)
\;=\;
\Eb\bigl[T(d)\mid T(d)\le Q_\gamma\{T(d)\}\bigr],
\qquad \gamma\in(0,1),
\end{equation}
which can be regarded as a one-sided trimmed mean, providing an analogue of classical robust statistical measures \citep{bickel1975descriptive, clarke2000note, lugosi2021robust}.} Maximizing $V_\gamma^1(d)$ therefore seeks a treatment rule that maximizes the average survival time within the worst-performing $\gamma$ proportion induced by that rule. Here, $\gamma$ is a prespecified tail-emphasis level rather than a data-driven tuning parameter. Smaller values of $\gamma$ place greater emphasis on the most unfavorable survival outcomes, whereas larger values gradually move the criterion toward the overall mean survival time.

The criterion $V_\gamma^1(d)$ admits a convenient representation in terms of conditional value-at-risk (CVaR)  \citep{rockafellar2000optimization}. For a real-valued random variable $Z$ with cumulative distribution function (CDF) $F$ and level $\eta\in(0,1)$, the value-at-risk is
\[
\mathrm{VaR}_\eta(Z)
=
\inf\{\alpha\in\mathbb R:F(\alpha) \ge \eta\},
\]
and the corresponding CVaR, or expected shortfall, is given by
\[
\mathrm{CVaR}_\eta(Z) = \frac{1}{1-\eta} \int_{z\geq \textrm{VaR}_{\eta}(Z)} z \, \textrm{d} F(z) 
=
\inf_{\alpha\in\mathbb R}
\left\{
\alpha+\frac{1}{1-\eta}\,\Eb[(Z-\alpha)_+]
\right\},
\]
where $(u)_+=\max(u,0)$. In risk management, $\textrm{VaR}_{\eta}(Z)$ often serves as a loss cutoff point and $\mathrm{CVaR}_\eta(Z)$ then represents the expected loss given that you breached $\textrm{VaR}_{\eta}(Z)$. Taking \(Z=-T(d)\) and \(\eta=1-\gamma\), and using the fact that \(\Pr[T(d)\le Q_\gamma\{T(d)\}]=\gamma\) under continuity of the distribution of \(T(d)\), we obtain
\begin{equation}\label{eq:VV}
V_\gamma^1(d)
=
\frac1\gamma\Eb\bigl[T(d)\,I\{T(d)\le Q_\gamma\{T(d)\}\}\bigr]
=
-\mathrm{CVaR}_{1-\gamma}(-T(d))
=
\sup_{\alpha\in\mathbb R}
\left\{
\alpha-\frac1\gamma\Eb\bigl[(\alpha-T(d))_+\bigr]
\right\}.
\end{equation}
Thus $V_\gamma^1(d)$ coincides with $-\mathrm{CVaR}_{1-\gamma}(-T(d))$. Based on this, it is convenient to express \eqref{eq:VV} under the idealized data generating process, that is,
\begin{equation*}
V_\gamma^1(d)
=
\sup_{\alpha\in\mathbb R}
\left\{
\alpha
-
\Eb\left[
\frac{I\{A=d(X)\}}{\gamma \cdot \pi(A\mid X)}
(\alpha-T)\,I\{T\le\alpha\}
\right]
\right\}.
\end{equation*}

With such a connection with CVaR, $V_\gamma^1(d)$ indeed defines a tail risk function specific to the treatment rules, as indicated in \cite{Qi2019EstimatingID}: it penalizes rules that produce very short survival times for a subset of patients, in the same spirit as CVaR, which is widely used in finance and reliability to control extreme losses. Therefore, we refer to $V_\gamma^1(d)$ as the CVaR criterion that aims to maximize the truncated mean survival time given a pre-specified survival probability. Next,
we propose another criterion that maximizes the survival probability at a quality-adjusted
survival time, where the quality-adjusted survival time is, in turn, determined by the truncated mean
survival time. In practice, we recommend users choose a criterion based on their goal.

\begin{remark}[Relation to restricted mean survival time]\label{rem:rmst_relation}
The criterion $V_\gamma^1(d)$ targets a quantile-indexed truncated mean and can be extended to quantile-indexed restricted mean survival time (RMST). 
For maximizing a fixed-horizon RMST, 
\cite{diaz2018targeted} develop a targeted maximum likelihood estimation framework and \cite{zhao2025efficient} propose a robust transfer learning approach to learn optimal treatment rules, respectively.

The RMST with respect to $T(d)$ at horizon $t_{0}>0$ is
\[
\mathrm{RMST}_{t_{0}}(d)
=
\Eb[\min\{T(d),t_{0}\}]
=
\Eb\bigl[T(d)\,I\{T(d)\le t_{0}\}\bigr]
+
t_{0}\Pr(T(d)>t_{0}).
\]
If we choose $t_{0}=Q_\gamma\{T(d)\}$, then
\[
\mathrm{RMST}_{Q_\gamma\{T(d)\}}(d)
=
\gamma V_\gamma^1(d)
+
(1-\gamma)\,Q_\gamma\{T(d)\}.
\]
Thus, a quantile-indexed RMST decomposes into a scaled truncated mean criterion plus a scaled quantile term. Under mild regularity conditions, $\mathrm{RMST}_{Q_\gamma\{T(d)\}}(d)$ admits an alternative representation 
\[(1-\gamma)\alpha^*(d) + \mathbb E[T(d)I\{T(d)\leq \alpha^*(d)\}], \]
where $\alpha^*(d)$ is the maximizer in \eqref{eq:VV}, that is,
$$\alpha^*(d) = \arg \max_{\alpha \in \mathbb R} \{\alpha - \frac1\gamma\mathbb E[(\alpha-T(d))I\{T(d)\leq \alpha\}]\}.$$
Therefore, our proposed method can be readily applied for optimal rule estimation under the criterion $\mathrm{RMST}_{Q_\gamma\{T(d)\}}(d)$; see the supplementary material for more details.
\end{remark}

\subsection{The buffered criterion: maximizing buffered survival probabilities}\label{subsec:2}

For many clinical applications, survival probabilities are often more directly interpretable, especially when decisions hinge on meeting a clinically meaningful milestone. A widely used estimand is the probability of surviving beyond a time point $t_0$, i.e., $\Pr(T(d)>t_0)$ \citep{gelber1989quality, jiang2017estimation}. A central practical difficulty, however, is that the choice of $t_0$ is often ambiguous: there may be multiple clinically reasonable cutoffs, yet even modest changes in $t_0$ can lead to materially different optimization targets and hence different optimal rules.

Note that maximizing $\Pr(T(d)>t_0)$ is equivalent to minimizing the fraction of individuals whose survival time falls below the cutoff, so the choice of cutoff implicitly defines what constitutes a vulnerable group. Rather than fixing $t_0$ a priori, we anchor the cutoff through a
user-specified lower-tail mean level $\tau$. Specifically, for each rule $d$, we let the
cutoff be determined by the smallest time point $q_\tau(d)$ such that the average survival time
among individuals with $T(d)\le q_\tau(d)$ equals $\tau$. In this way, $q_\tau(d)$ serves as a buffered cutoff and conservative boundary for vulnerable patients, and it can be less sensitive compared to thresholding at $t_0$ in settings with a pronounced lower tail.

To achieve this goal, we define the buffered survival probability as below,
\begin{equation}
    \label{eq:VVVV}
    V^2_{\tau}(d)\equiv \begin{cases} 1- \gamma, \,\,  \mbox{ with } \gamma \mbox{ satisfying } V_{\gamma}^1(d) 
    = \tau & \mbox{ for }  \tau \in (\inf(T(d)), \mathbb{E}[T(d)] \, )\\
    0 &  \mbox{ for } \tau \geq \mathbb{E}[T(d)] \\
    1 &  \mbox{ for } \tau \leq  \inf(T(d)).
    \end{cases}
\end{equation}
Such definition is proper because $V_{\gamma}^1(d) = -\mbox{CVaR}_{1-\gamma}(-T(d)) $, which is a continuous and strictly increasing function with respect to $\gamma$ for $\gamma \in (0, 1)$ with the range $(\inf(T(d)), \mathbb{E}[T(d)] \,)$.

\begin{remark}
When the survival time $T(d)$ is a continuous random variable, the value function reduces to $V^2_\tau(d) = \Pr(T(d) > q_\tau(d))$  with $q_\tau(d)$ satisfying $\Eb[T(d)\,|\,T(d)\leq q_\tau(d)]=\tau$.
Therefore, Equation~\eqref{eq:VVVV} enjoys the interpretation of the survival function at the given time point $q_\tau(d)$. So maximizing $V^2_\tau(d)$ is equivalent to maximizing the surviving proportion beyond a truncation while keeping the mean survival among those below that truncation at the specified level. Here $\tau$ is a user-specified level that encodes the desired average performance among vulnerable patients.
\end{remark}

\st{We refer to $V^2_\tau(d)$ as the buffered criterion because it is closely related to the bPOE introduced by \citet{rockafellar2010buffered} as a structural reliability measure in the design and optimization of structures.  The mathematical properties of the bPOE and its applications in finance, reliability, and machine learning are subsequently studied in \cite{mafusalov2018buffered, norton2019maximization, rockafellar2020minimizing}.}

From a probabilistic perspective, the survival function
is known as the probability of exceedance (POE) as a measure of reliability and seems appealing
due to its straightforward interpretation. However, the POE disregards the magnitude of outcomes in the tail, and thus, a criterion that directly maximizes the survival function may lead to
treatment rules that take arbitrarily bad outcomes with a substantial probability. Moreover, the
optimization of POE may encounter significant computational difficulties in practice as the POE
of a random function is not continuously differentiable or even discontinuous in many situations,
particularly when there are only finite scenarios. The theoretical and computational challenges
of POE are discussed in detail in \citet{rockafellar2010buffered}. To overcome those difficulties, bPOE is proposed in \citet{rockafellar2010buffered} as an alternative measure of reliability, which provides an upper bound of POE with significant computational advantages.

In the following lemma, we formally link the proposed framework to the concept of bPOE
\citep{mafusalov2018buffered}, so that we can effectively leverage Proposition 2.2 of \citet{mafusalov2018buffered} to solve our optimization problem of interest.
\begin{lemma}\label{lemma:maxmin}
Let $V^2_\tau(d)$ be defined as in \eqref{eq:VVVV} with $\tau\in(\inf T(d),\Eb[T(d)]]$. Then
\[
V^2_{\tau}(d)
=
1 - \inf_{c\ge 0} \Eb\!\left[ \max\bigl\{0,\,c(-T(d)+\tau)+1\bigr\} \right].
\]
Consequently, maximizing $V^2_\tau(d)$ over $d$ is equivalent to minimizing the following over $d$,
\[
M^2_\tau(d)
=
\inf_{c\ge 0} \Eb\!\left[ \max\bigl\{0,\,c(-T(d)+\tau)+1\bigr\} \right].
\]
\end{lemma}

Based on the lemma, it is also convenient to express the value function according to the idealized data generating process, i.e.,
\begin{equation*}
M_\tau^2(d)
=
\inf_{c\ge 0} \Eb\!\left[\frac{I\{A=d(X)\}}{\pi(A\mid X)}\max\bigl\{0,\,c(-T+\tau)+1\bigr\} \right].
\end{equation*}

\subsection{Accounting for censoring: identification and learning}\label{subsec:rc}

In practice, survival time is typically subject to right censoring, including administrative censoring.
In such settings, the full mean survival time $\mathbb{E}[T]$ is generally not identifiable without additional structure.
We therefore develop identification and estimation of our proposed criteria within a standard right-censoring framework. 
Specifically, we observe \[
Z = (X,A,Y,\Delta), \qquad Y = T \wedge C,\quad \Delta = I(T \le C),
\]
where $C$ denotes the censoring time, $Y$ is the observed follow-up time, and $\Delta$ is the event indicator.

Throughout this subsection and thereafter, we impose the following additional assumptions \citep{fleming2011counting,cui2020estimating}:

\textbf{(A2.1)} \textit{Bounded horizon.} The survival time $T$ admits a finite maximal horizon $0<h<\infty$.

\textbf{(A2.2)} \textit{Conditional independent censoring.} $T\perp C\mid (X,A)$.

\textbf{(A2.3)} \textit{Positivity.} There exists $0<\eta_C\le 1$ such that $\Pr(C\ge h\mid X,A)\ge \eta_C$.

\noindent Under \textbf{(A2.1)}–\textbf{(A2.3)}, the robust criteria $V_\gamma^1(d)$ and $M_\tau^2(d)$ can be identified from the right-censored data $\{(X_i,A_i,Y_i,\Delta_i)\}_{i=1}^n$.

To streamline notation, we write $V_\gamma^1(d)$ and $M_\tau^2(d)$ as $V^1(d)$ and $M^2(d)$, respectively, with the dependence on $\gamma\in(0,1)$ and $\tau\in(\inf T(d), \Eb[T(d)]]$ understood. Because $T$ is only partially observed, we first express $V^1(d)$ and $M^2(d)$ in terms of the observed variables $Z$ and the conditional survival function of the censoring distribution, that is, 
\(
S_C(t\mid X,A) = \Pr(C \ge t \mid X,A).
\)

\begin{lemma}\label{lemma:1}
Under assumptions \textbf{(A2.1)} to \textbf{(A2.3)}, we have
\begin{align*}
V^1(d)
&=
\sup_{\alpha\in\mathbb R}
\left\{
\alpha
-
\frac1\gamma\Eb\left[
\frac{\Delta\,(\alpha-Y)I(\alpha\ge Y)}{S_C(Y\mid X,A)}
\frac{I\{A = d(X)\}}{\pi(A\mid X)}
\right]
\right\},\\[0.1in]
M^2(d)
&=
\inf_{c\ge 0}
\Eb\left[
\frac{\Delta\,\max\{0,\,c(-Y+\tau)+1\}}{S_C(Y\mid X,A)}
\frac{I\{A = d(X)\}}{\pi(A\mid X)}
\right].
\end{align*}
\end{lemma}

Lemma~\ref{lemma:1}  provides the explicit formulation of the two robust criteria with censored data and forms the basis for our learning procedure. To turn these population functions into estimators using the observed data $\{(X_i,A_i,Y_i,\Delta_i)\}_{i=1}^n$ by replacing expectations with empirical averages, two main challenges arise.

First, the censoring survival function $S_C(Y| X,A)$ is unknown in practice and must be estimated. A natural approach is to fit a working model for the censoring distribution, for example via a Cox proportional hazards model \citep{cox1972regression} or random survival forests \citep{ishwaran2008}, and to plug the resulting estimator $\widehat S_C(Y|X,A)$ into the expressions in Lemma~\ref{lemma:1}.

Second, the indicators $I\{A=d(X)\}$ appearing in $V^1(d)$ and $M^2(d)$ lead to non-smooth, combinatorial optimization over treatment rules. To obtain implementable learning objectives, we write $d(X)=\operatorname{sign}\{f(X)\}$ for some function $f\in\mathcal H$ and approximate $I\{A=\operatorname{sign}(f(X))\}=I\{A f(X)>0\}$ by a smooth surrogate loss $\mathcal L\{A f(X)\}$. Following \citet{Qi2019EstimatingID}, we consider surrogate losses of the form $\mathcal L(t) = \mathcal L_1(t) - \mathcal L_2(t)$, where $\mathcal L_1$ and $\mathcal L_2$ are differentiable convex functions. For concreteness, one convenient choice is
\[
\cL(u) =
\begin{cases}
0, & u < -\delta, \\
\frac{1}{2} (1+u/\delta)^2, & -\delta \le u < 0, \\
1 - \frac{1}{2}(1-u/\delta)^2, & 0 \le u < \delta,\\
1, & u \ge \delta,
\end{cases}
\]
with
\[
\cL_1(u) =
\begin{cases}
0, & u \le -\delta, \\
\frac{1}{2}(1+u/\delta)^2, & -\delta < u \le 0, \\
\frac{1}{2} + u/\delta, & u > 0,
\end{cases}
\qquad
\cL_2(u) =
\begin{cases}
0, & u \le 0, \\
\frac{1}{2} (u/\delta)^2, & 0 < u \le \delta, \\
u/\delta - \frac{1}{2}, & u > \delta.
\end{cases}
\]

Combining Lemma~\ref{lemma:1}, the estimator $\wh S_C$, and the smooth surrogate loss $\mathcal L$, we arrive at the following learning objectives:
\begin{align}
\label{eq:estimated_treatment_regime_CVaR}
\max_{f \in \mathcal H} \quad V^{\wedge}_\cL(f)
&= \sup_{\alpha\in\mathbb R} \left\{\alpha
-
\frac 1\gamma\Eb\left[
\frac{\Delta\,(\alpha-Y)I(\alpha\ge Y)}{\wh S_C(Y\mid X,A)}
\frac{\cL\{A f(X)\}}{\pi(A\mid X)}
\right]\right\},\\[0.1in]
\label{eq:estimated_treatment_regime_buffered}
\min_{f \in \mathcal H}  \quad  M^{\wedge}_\cL(f)
&=
\inf_{c \ge 0}
\Eb\left[
\frac{\Delta\,\max\{0,\,c(-Y+\tau)+1\}}{\wh S_C(Y\mid X,A)}
\frac{\mathcal L\{A f(X)\}}{\pi(A\mid X)}
\right].
\end{align}
In practice, with observed data $\{(X_i,A_i,Y_i,\Delta_i)\}_{i=1}^n$, we can then replace the expectations in \eqref{eq:estimated_treatment_regime_CVaR}–\eqref{eq:estimated_treatment_regime_buffered} by empirical averages and optimize them using DC programming as in \citet{Qi2019EstimatingID}. The resulting function $\hat f$ induces the estimated treatment rule $\hat d(x)=\operatorname{sign}\{\hat f(x)\}$. We develop the corresponding algorithmic implementation in the next section.

\section{Sampling-based learning algorithm}
\label{sec:algo}
In this section, we focus on the algorithmic development for learning the treatment rules under the two robust criteria with censored data $Z_i \equiv (X_i, A_i, Y_i, \Delta_i)$ for ${i=1, \ldots, n}$. 
We make two assumptions below.

\textbf{(A3.1)}  The surrogate function $\mathcal L(\bullet)$ satisfies that $\mathcal L(t) \in [0,1]$ for every $t \in \mathbb R$, and $ \mathcal L(\bullet) = \mathcal L_1(\bullet) - \mathcal L_2(\bullet)$  where $ \mathcal L_1(\bullet)$ and $
\mathcal L_2(\bullet)$ are smooth convex functions with the Lipschitz gradient modulus $L$; 

\textbf{(A3.2)}   $\mathcal H$ is a Reproducing Kernel Hilbert Space (RKHS)    with the norm $\| \bullet\|_{\mathcal H}$ and the corresponding real-valued kernel  $   k:{\mathcal {X}}\times {\mathcal {X}}\to \mathbb {R} $ that is symmetric and positive-definite.  


By the Representer theorem, an optimal solution admits the representation
$f(x) = \sum_{j=1}^n \beta_j k(x,X_j)$ for some coefficient vector
$\beta = (\beta_1,\ldots,\beta_n)^\top$, and if
$K_i \equiv (k(X_i,X_1),\ldots,k(X_i,X_n))^\top$, we have $f(X_i) = K_i^\top \beta$.
For treatment learning tasks under the two robust criteria, we can then formulate the corresponding empirical problems with regularization as follows:
\begin{align}
\underset{f \in \cH}{\mbox{max}} \quad \left( \sup_{\alpha \in \mathbb{R}} \,\,  \alpha  -  \frac{1}{n}\sum_{i = 1}^{n} \frac{\mathcal L(A_i    \, K_i^\top \beta)}{\pi(A_i | \Xb_i)}   \frac{\Delta_i (\alpha-Y_i)I(Y_i\leq \alpha)}{ \gamma \cdot \widehat S_C(Y_i|X_i,A_i)} - \lambda \|f\|_{\mathcal H}^2\right) ,\label{eq:ev1}\\[0.1in]
\underset{f \in \cH}{\mbox{min}} \quad \left( \inf_{c \geq 0} \,\, \frac{1}{n}\sum_{i = 1}^{n} \frac{ \mathcal L(A_i    \, K_i^\top \beta)}{\pi(A_i | \Xb_i)} \left\{\frac{\Delta_i \max(0, c(-Y_i+\tau)+1)}{\widehat S_C(Y_i|X_i,A_i)}\right\} + \lambda \|f\|_{\mathcal H}^2 \right) .\label{eq:ev2}
\end{align}
Hereinafter, Problems~\eqref{eq:ev1} and \eqref{eq:ev2} are the target optimization problems for the CVaR and the buffered criteria, respectively, and the rest of this section is about how to solve them efficiently in a scalable way. We focus on solving Problem~\eqref{eq:ev1}, and we will show later that the proposed algorithm can be applied to solve Problem~\eqref{eq:ev2} as well. 

\subsection{Empirical DC formulation}
\label{sec:edc}
With $W_i \equiv  \frac{\Delta_i}{\pi(A_i \mid X_i)\, \wh S_C(Y_i | X_i, A_i)}$, $K_{ij} \equiv  k(X_i, X_j)$,  Problem \eqref{eq:ev1} can be reformulated as follows,
\begin{equation}
\label{eq:V1_opt_surrogate} 
\underset{\beta \in \mathbb{R}^n}{\mbox{min}}  \quad  \left(  V(\beta) \equiv \underset{\alpha \in \mathbb R}{\inf} \, -\frac{1}{n} \sum_{i=1}^n \,  {\psi}(\beta, \alpha; Z_i) +  \lambda \, \beta^\top K \beta\right), 
\end{equation}  
where $\psi(\beta, \alpha; Z_i) \equiv \alpha -  \frac1\gamma\left(\,\mathcal L_1\big(A_i \, K_i^\top \beta \big) -\mathcal L_2\big(A_i \,  K_i^\top \beta \big)\,\right) \, W_i\,  (\alpha - Y_i)_+$. To simplify notation, we henceforth omit the superscripts and subscripts on $V$.
According to \cite{rockafellar2000optimization}, the optimal solution of $\alpha$ in \eqref{eq:V1_opt_surrogate} should be contained in a bounded set. Moreover, since $ {\psi}(\beta, \bullet; Z_i)$ is a concave piecewise affine function with knots $\{Y_i\}_{i=1}^n$, it is shown in \cite{Qi2019EstimatingID} that the inner minimization over $\alpha$ in \eqref{eq:V1_opt_surrogate} achieves its optimum at one of the knots $\{Y_i\}_{i=1}^n$. Thus, the objective function of  \eqref{eq:V1_opt_surrogate} can be rewritten as 
\begin{equation*}
\begin{array}{lll}
V(\beta)  =  \underset{ j \in [n] }{\mbox{min}} \, -\displaystyle{\frac{1}{n} \sum_{i=1}^n } \,  {\psi}(\beta, Y_j; Z_i) +  \lambda \, \beta^\top K \beta,
\end{array}
\end{equation*}
and with
\(
\overline Y_i \equiv \sum_{k=1}^n (Y_k - Y_i)_+,
\widehat Y_{i,j} \equiv \overline Y_i - (Y_j - Y_i)_+
= \sum_{k\ne j} (Y_k - Y_i)_+
\), it can further be decomposed into the difference of two convex functions as follows:
\begin{equation}
\label{eq:dc_decomposition}
\begin{array}{lll}
 V(\beta) & =  & \underset{ j \in [n] }{\mbox{min}} \quad - \displaystyle{\frac{1}{n}} \sum_{i=1}^n  \Big( Y_j  -  \frac1\gamma\big(\,\mathcal L_1(A_i  \, K_i^\top \beta) -\mathcal L_2(A_i  \, K_i^\top \beta)\,\big) (Y_j - Y_i)_+ W_i \Big) +  \lambda \,  \beta^\top K \beta  \\[0.1in]
&= &  \underbrace{\displaystyle{\frac{1}{n\gamma}}     \sum_{i=1}^n \,     \mathcal L_1(A_i  \, K_i^\top \beta)\,  \overline Y_i W_i }_{\mbox{denoted by } \varphi_{1}(\beta)} +  \lambda \, \beta^\top K \beta  \\[0.5in] 
&& -  \, \underset{ j \in [n] }{\mbox{max}} \, \underbrace{  \displaystyle{\frac{1}{n\gamma}} \sum_{i=1}^n  \left( \,Y_j \, \gamma +   \displaystyle{{\mathcal L_1(A_i \,  \, K_i^\top \beta)}} \, \wh Y_{i,j} \, W_i  +   \mathcal L_2(A_i \,  \, K_i^\top \beta) (Y_j - Y_i)_+ W_i \right) \,  }_{\mbox{denoted by } \varphi_{2,j}(\beta)}.
\end{array}
\end{equation}
Then the optimization problem \eqref{eq:V1_opt_surrogate} is a DC program, and as implemented in \cite{Qi2019EstimatingID}, the deterministic DC algorithm and its enhanced variant proposed by \cite{pang2016computing} can be applied for obtaining a critical or d-stationary solution. 

\subsection{Sampling-based DCA}
\label{sec:sdc}
However, due to the finite-sum structures of the two convex component functions in \eqref{eq:dc_decomposition},  DC algorithms need to iteratively solve convex subproblems with finite-sum over the entire data set, and thus could be highly computationally inefficient especially when the data size $n$ is large. Therefore, we propose a sampling-based DC algorithm for solving Problem \eqref{eq:V1_opt_surrogate}. Specifically, let $I_\nu$ be a multiset of indices obtained from
$N_\nu$ independent uniform draws with replacement from $[n]$, where $N_\nu=|I_\nu|$ counts multiplicities. Throughout, sums over $I_\nu$ count each index according to its multiplicity. Using the corresponding sample $\{Z_i:i\in I_\nu\}$, we
construct the sampling-based approximation:
\[
\begin{array}{lll}
 \wt V^{(\nu)}(\beta)& \equiv  & \underset{ j \in I_\nu }{\mbox{min}} \, - \displaystyle{\frac{1}{N_\nu}} \sum_{i \in I_\nu}  \,  {\psi}(\beta, Y_j; Z_i) + \lambda \, \beta^\top K \beta  \\[0.3in]
& = & \underbrace{\displaystyle{\frac{1}{N_\nu \gamma}}      \sum_{i \in I_\nu} \,     \mathcal L_1(A_i  \, K_i^\top \beta)\,   \overline Y_i^{(\nu)} \, W_i }_{\mbox{denoted by } \tilde{\varphi}_{1}^{(\nu)}(\beta)} + \lambda \, \beta^\top K \beta   \\[0.5in] 
&  &-   \underset{ j \in I_\nu }{\mbox{max}} \quad \underbrace{  \displaystyle{\frac{1}{N_\nu \gamma}} \sum_{i \in I_\nu}  \left( \,Y_j \, \gamma +    \displaystyle{{\mathcal L_1(A_i \,  \, K_i^\top \beta)}}   \wh Y_{i,j}^{(\nu)} \, W_i  +   \mathcal L_2(A_i \,  \, K_i^\top \beta) (Y_j - Y_i)_+ W_i \right) \,  }_{\mbox{denoted by } \tilde \varphi^{(\nu)}_{2,j}(\beta)},
\end{array}\]  
where
\(
\overline{Y}_i^{(\nu)}
\equiv \sum_{k\in I_\nu}(Y_k-Y_i)_+,
\widehat{Y}_{i,j}^{(\nu)}
\equiv \overline{Y}_i^{(\nu)}-(Y_j-Y_i)_+.
\)
Due to the finite min operator, such a sampling-based approximation function $\wt{V}^{(\nu)}(\beta)$  is a biased estimator of $V(\beta)$, which necessitates a particular sampling scheme in the follow-up algorithmic development and convergence analysis. With the following active index sets,
 \[
\begin{array}{ll}
\mathcal{M}(\beta) \equiv   \underset{j \in [n]}{\mbox{argmin}} \, \,\, \  -\displaystyle{\frac{1}{n} \sum_{i=1}^{n}}   \psi(\beta, Y_j; Z_i) ,\\[0.2in] 
\widetilde{\mathcal M}_{\varepsilon}^{(\nu)}(\beta)
\equiv
\left\{
\ell\in I_\nu:
-\frac{1}{N_\nu}
\sum_{i\in I_\nu}\psi(\beta,Y_\ell;Z_i)
\le
\min_{j\in I_\nu}
\left[
-\frac{1}{N_\nu}
\sum_{i\in I_\nu}\psi(\beta,Y_j;Z_i)
\right]
+\varepsilon
\right\}\  \text{with} \ \varepsilon\ge 0, 
\end{array}
\]
we define the linear approximation functions of ${\varphi}_{2,j}(\beta)$ and $\wt{\varphi}_{2,j}^{(\nu)}(\beta)$ respectively at a reference point $\beta^{(\nu)}$:
\begin{equation}
\label{eq:linearization}
\begin{array}{ll}
\widehat{\varphi}_{2,j}(\beta; \beta^{(\nu)}) \equiv & {\varphi}_{2,j}(\beta^{(\nu)}) + \langle \nabla {\varphi}_{2,j}(\beta^{(\nu)}), \, \beta -\beta^{(\nu)}\rangle, \mbox{ for any } j \in \mathcal M(\beta^{(\nu)}),\\[0.1in] 
\widehat{\varphi}_{2,j}^{(\nu)}(\beta; \beta^{(\nu)}) \equiv &\wt{\varphi}^{(\nu)}_{2,j}(\beta^{(\nu)}) + \langle \nabla \wt{\varphi}^{(\nu)}_{2,j}(\beta^{(\nu)}),\, \beta -\beta^{(\nu)} \rangle, \mbox{ for any } j \in \wt{\mathcal{M}}^{(\nu)}_{\varepsilon}(\beta^{(\nu)}).
\end{array}
\end{equation}
 Accordingly, we can construct the sampling-based convexified approximation function of $V(\beta)$ as follows,  $
\widehat V^{(\nu)}_{j}(\beta; \beta^{(\nu)}) \equiv \wt \varphi_1^{(\nu)}(\beta) - \widehat{\varphi}_{2,j}^{(\nu)}(\beta; \beta^{(\nu)})+ \lambda \, \beta^\top K \beta $ for every $j \in \wt{\mathcal{M}}^{(\nu)}_{\varepsilon}(\beta^{(\nu)})$.

Building upon these, we develop a sampling-based algorithm below which is embedded with three \textit{sequential} steps: sampling, DC-decomposition, and convexification. At iteration $\nu$, with i.i.d. generated sample set $\{Z_i: i \in I_{\nu}\}$ containing newly generated samples and all historically generated samples, for every index $j \in \wt{\mathcal{M}}^{(\nu)}_{\varepsilon}(\beta^{(\nu)})$, we can constitute a sampling-based convexified approximation function $\widehat V^{(\nu)}_{j}(\beta; \beta^{(\nu)})$, and further obtain the corresponding proximal mapping point by implementing the convex-programming solver. Due to the finite-max structure of the second component function in DC decomposition \eqref{eq:dc_decomposition}  together with the biased estimation, we tailor the enhancement technique of deterministic DCA in \cite{pang2016computing} for such a sampling-based algorithm; namely in Steps 4 and 5 of the proposed algorithm, we select the best candidate solution over the $\varepsilon$-active index set as the next iterate point. We will show such an enhancement yields convergence to a directional stationary solution of the original DC problem, a particularly sharp stationarity type for general nonconvex programs.

\noindent\makebox[\linewidth]{\rule{6.5in}{1pt}}
\vspace{-0.3in}

 \noindent {\bf Algorithm 1}\\
\vspace{-0.4in}

\noindent\makebox[\linewidth]{\rule{6.5in}{1pt}}
\vspace{-0.3in}

 \begin{algorithmic}[1]
     \STATE \textbf{Initialization:} Let $\beta^{(1)} \in \mathbb R^n$ and positive scalars
    $\rho, \lambda, \varepsilon$ be given. Initialize $I_0$ as an empty list and set $N_0=0$.
    \FOR {$\nu=1,2, \cdots$,}
    \STATE Generate an independent batch of $\Delta_\nu$ i.i.d. draws by sampling uniformly with replacement from $\{Z_i\}_{i\in[n]}$.
    Let $K_\nu$ denote the multiset of sampled indices, and set $I_\nu\equiv I_{\nu-1}\uplus K_\nu$ and $N_\nu\equiv N_{\nu-1}+\Delta_\nu$,
    where $\uplus$ denotes multiset union, retaining repeated indices;
    \STATE For every $ j \in \wt{\mathcal{M}}^{(\nu)}_{\varepsilon}(\beta^{(\nu)})$, compute 
$ \beta^{(\nu+1/2)}_j = \underset{\beta \, \in \, \mathbb R^n}{\mbox{argmin}} \left\{ \,		
 \widehat V^{(\nu)}_{j}(\beta; \beta^{(\nu)}) 
 + \displaystyle{
        \frac{1}{2 \rho}
    } \, \| \beta - \beta^{(\nu)}\|^2 \, \, \right\};$
    \STATE Let $j_\nu \in \displaystyle{\mbox{argmin}_{ j \in \wt{\mathcal{M}}^{(\nu)}_{\varepsilon}(\beta^{(\nu)})}} \, \left\{  \widehat V_j^{(\nu)}( \beta^{(\nu+1/2)}_j;\beta^{(\nu)}) + \displaystyle{
        \frac{1}{2 \rho}
    } \, \| \beta^{(\nu+1/2)}_j - \beta^{(\nu)}\|^2 \, \, \right\} $ and set $\beta^{(\nu+1)} = \beta^{(\nu+1/2)}_{j_\nu}$.
    \ENDFOR
\end{algorithmic}
\noindent\makebox[\linewidth]{\rule{6.5in}{1pt}}

It is worth noticing that the proposed sampling-based DC algorithm has several major distinctions from some recent works on sampling-based DC algorithms in \cite{an2019stochastic} and \cite{liu2022risk}, which thus necessitates a separate convergence analysis. First, the stochastic DC algorithm of \cite{an2019stochastic} is developed for a class of DC programs where both convex component functions are expectation functions, whereas the DC programs studied in the present paper and in \cite{liu2022risk} both contain the second component functions, which are the value functions of finite-sum or expectation functions. Second, we construct the sampling-based convexified approximation function with sequential sampling and DC-decomposition,  so that we only use a single sample set for both convex component functions. Since the two component functions are likely to have positive covariance due to the shared parts, such a common random number technique  leads to  solution sequences with less volatility compared to the algorithm in \cite{liu2022risk} with independent sampling strategy for the two component functions.  The third distinction is that based on the finite-max structure of \eqref{eq:dc_decomposition}, we integrate the enhanced technique into the sampling-based DC algorithm in order to obtain the directional stationary solution, which distinguishes the proposed algorithm and the convergence analysis from the sampling-based algorithm in \cite{liu2022risk} which can only achieve the convergence to critical solutions to the stochastic DC value function optimization problem.

Recall that for the minimization of a B(ouligand)-differentiable function $f$ over a convex set $\mathcal{X}'$, we say that $\bar x$ is a directional stationary solution if $f'(\bar x; x-\bar x) \equiv \lim_{t \to 0} \frac{f(\bar x+ t(x- \bar x)) - f(\bar x)}{t} \geq 0$ for any $x\in \mathcal{X}'$. We present the convergence result of the proposed algorithm with the proof provided in the supplementary material.

	\begin{theorem} \label{thm:convergence_sdc} \rm
Under assumptions \textbf{(A3.1)} and \textbf{(A3.2)}, with $\Delta_\nu \geq  C$ for some positive integer $C$, with arbitrary $\rho >  0$ and $\lambda >0$, 
		any limit point $\beta^{\,\infty}$ of the sequence $\{\beta^{(\nu)}\}$ generated by Algorithm 1, if one exists, is a directional stationary point of \eqref{eq:V1_opt_surrogate} with probability 1.
		\end{theorem}

\subsection{Extension to the buffered criterion}
For learning the treatment rule under the buffered criterion, the optimization problem can be formulated below,   
\begin{align*}
\underset{\beta \in \mathbb R^n}{\mbox{min}} \quad  \inf_{c \geq 0} \, \, \frac{1}{n}\sum_{i = 1}^{n} \left(\mathcal L_1(A_i  \, K_i^\top \beta) - \mathcal L_2(A_i  \, K_i^\top \beta) \right) \max(0, c(-Y_i+\tau)+1) W_i + \lambda \beta^\top K \beta.
\end{align*}
Similar to the case of the CVaR criterion, the inner minimization over $c$ also achieves the optimum at one of the knots $\mathcal C_n \equiv \{0\} \cup \{\frac{1}{Y_i - \tau}:Y_i>\tau\}_{i=1}^n$. Hence, the optimization problem can be written as 
\begin{align}
\label{eq:buffered_opt}
\underset{\beta \in \mathbb R^n}{\mbox{minimize}}\,\,  M(\beta) \equiv \min_{c \in \mathcal C_n} \,\frac{1}{n}\sum_{i = 1}^{n} \left(\mathcal L_1(A_i  \, K_i^\top \beta) - \mathcal L_2(A_i  \, K_i^\top \beta) \right) \max( 0, c(-Y_i+\tau) +1 ) W_i + \lambda \beta^\top K \beta.
\end{align}
The objective function $M(\beta)$ is the minimization of a finite sum of DC functions, and therefore has the DC decomposition similar to that in \eqref{eq:dc_decomposition}. Hence, to learn the treatment rule under the buffered criterion, we can also leverage Algorithm~1 to solve the optimization problem \eqref{eq:buffered_opt}.

\section{Asymptotic properties}\label{sec:theory}
In this section, we provide the asymptotic results of the proposed estimated treatment rules with respect to \eqref{eq:ev1} and \eqref{eq:ev2}. The purpose of this section is to demonstrate that the value function of the estimated optimal treatment rule converges to the optimal value function under the CVaR and the buffered criteria, respectively. In particular, we establish Fisher consistency, excess risk bound, and universal consistency of the estimated optimal treatment rules.

We start with introducing some notation for the CVaR criterion. Recall that according to Lemma \ref{lemma:1}, 
\begin{align*}
V^1(d) &=\sup_{ \alpha \in \mathbb{R}} \left\{ \Eb\left[ \left\{\alpha - \frac{\Delta (\alpha-Y)I(\alpha\geq Y) }{\gamma\cdot S_C(Y|X,A)} \right\} \frac{I\{A=d(X)\}}{\pi(A|X)}\right] \right\}.
\end{align*}
We denote the expectation inside the supremum by \(V^1(d,\alpha)\).
With $d(X) = \mbox{sign}(f(X))$ and the surrogate function $\mathcal L$ of the indicator function, we define
\begin{align*}
V^{1}_\cL(f, \alpha) \equiv & \Eb\left[ \left\{\alpha - \frac{\Delta (\alpha-Y)I(\alpha\geq Y) }{\gamma\cdot S_C(Y|X,A)} \right\} \frac{\cL(Af(\Xb))}{\pi(A | \Xb)} \right],
\end{align*}
as the surrogate value function. We define the value function and surrogate value function corresponding to the working model as
\begin{align*}
V^{\wedge, 1}(d, \alpha) & \equiv  \Eb\left[ \left\{\alpha  - \frac{\Delta (\alpha-Y)I(\alpha\geq Y) }{ \gamma\cdot \widehat S_C(Y|X,A)} \right\} \frac{I\{A = d(X)\}}{\pi(A|X)}\right],\\
V^{\wedge,1}_\cL(f, \alpha) & \equiv  \Eb \left[ \left\{\alpha - \frac{\Delta (\alpha-Y)I(\alpha\geq Y) }{\gamma\cdot \widehat S_C(Y|X,A)} \right\} \frac{\cL(Af(\Xb))}{\pi(A|X)}\right],
\end{align*}
respectively.  We also define 
\begin{align*}
V'^{,1}(d, \alpha) &\equiv \Eb\left[ \left\{\alpha  - \frac{\Delta (\alpha-Y)I(\alpha\geq Y) }{ \gamma\cdot \widetilde S_C(Y|X,A)} \right\} \frac{I\{A = d(X)\}}{\pi(A|X)}\right],\\
V'^{,1}_\cL(f, \alpha) &\equiv \Eb \left[ \left\{\alpha  - \frac{\Delta (\alpha-Y)I(\alpha\geq Y) }{\gamma \cdot \widetilde S_C(Y|X,A)} \right\} \frac{\cL(Af(\Xb))}{\pi(A|X)}\right],
\end{align*}
where 
$\widetilde S_C(\cdot|X,A)$ is the probability limit of $\widehat S_C(\cdot|X,A)$. Here and below, expectations in the ‘hat’ and ‘prime’ functions are taken with respect to the distribution of $(X,A,Y,\Delta)$, treating $\widehat S_C$ (and hence $\widetilde S_C$) as fixed. In addition, since $0\leq T\leq h$, we restrict the auxiliary
parameter $\alpha$ to $[0,h]$, which contains an optimal threshold for the population CVaR criterion, and that the buffered population surrogate optimum is attained at some finite $c\geq0$.
We also assume that $\widetilde S_C(t|x,a)$ is uniformly bounded away from zero over $[0,h]\times\cX\times\{-1,1\}$. Meanwhile, we let $\cF$ denote the class of all measurable real-valued functions, and denote
\(
(\widetilde f,\widetilde\alpha)
\in \underset{f\in\cF,\,\alpha\in[0,h]}{\arg\max}
V'^{,1}_{\cL}(f,\alpha)\).

We first establish Fisher consistency of estimating optimal treatment rules under $V'^{,1}_\cL(f, \alpha)$ to justify the use of the surrogate loss $\cL(u)$.
Theorem~\ref{1:fisher} shows that optimizing the surrogate objective yields an optimal rule for the original criterion.
\begin{theorem}\label{1:fisher}
For $(\widetilde f,\widetilde\alpha)$ defined above,
$(\sign(\widetilde f),\widetilde\alpha)$ maximizes
$V'^{,1}(d,\alpha)$ over all measurable treatment rules $d$
and $\alpha\in[0,h]$.
\end{theorem}

We next study the consistency of the estimated treatment rule. The analysis accounts for errors in estimating the conditional censoring survival functions and in learning the treatment rule.
As an intermediate step, we establish the following excess risk bound. 
\begin{theorem}\label{1:excess}
For any $f\in\cF$ and $\alpha\in[0,h]$,
\begin{align*}
\sup_{f'\in\cF,\,\alpha'\in[0,h]}
V'^{,1}(\sign(f'),\alpha')
-
V'^{,1}(\sign(f),\alpha)\leq
2\left\{
V'^{,1}_{\cL}(\widetilde f,\widetilde\alpha)
-
V'^{,1}_{\cL}(f,\alpha)
\right\}.
\end{align*}
\end{theorem}
\noindent
By leveraging Theorem~\ref{1:excess}, we establish the following consistency result.
\begin{theorem}\label{1:limit}
Suppose that \textbf{(A2.1)}--\textbf{(A2.3)} hold,
the kernel is bounded, and $\calH$ is dense in $L^1(P_X)$,
where $P_X$ denotes the marginal distribution of $X$.
Assume that
\(
\inf_{x\in\cX,\,a\in\{-1,1\}}\pi(a|x)>0,
\)
and let $\lambda=\lambda_n>0$ satisfy
$\lambda_n\to0$ and $n\lambda_n\to\infty$.
Suppose also that
\(
\sup_{x\in\cX,\,0\leq t\leq h}
|\widehat S_C(t|x,a)-S_C(t|x,a)|
\xrightarrow{p}0
\)
for both $a=-1,1$.
Let $(\widehat f,\widehat\alpha)$ be a global solution
of \eqref{eq:ev1}, with $\lambda=\lambda_n$ and
$\alpha\in[0,h]$. Then
\[
\sup_{\alpha\in[0,h]}
V^1(\sign(\widehat f),\alpha)
\xrightarrow{p}
\sup_{f\in\calH,\,\alpha\in[0,h]}
V^1(\sign(f),\alpha).
\]
\end{theorem}

The assumption on $\widehat S_C(t|x,a)$ is quite general. In particular, the consistency with certain rates can be achieved by parametric models as well as nonparametric methods such as survival forests \citep{cuiy2017some} and nonparametric kernel smoothing methods \citep{sun2019}.
The rate of convergence of the estimated treatment rule can be studied under certain standard assumptions in the learning theory literature \citep{steinwart2008support}, and we omit the details here.

Analogously, for the buffered criterion, we define
\begin{align*}
V^{2}(d, c) & \equiv  \Eb\left[- \left\{\frac{\Delta \max(0, c(-Y+\tau)+1)}{S_C(Y|X,A)}  \right\} \frac{I\{A = d(X)\}}{\pi(A|X)}\right],\\
V^{2}_\cL(f, c) & \equiv  \Eb\left[- \left\{\frac{\Delta \max(0, c(-Y+\tau)+1)}{S_C(Y|X,A)}  \right\} \frac{\cL(Af(\Xb))}{\pi(A | \Xb)} \right],\\
V^{\wedge, 2}(d, c) & \equiv  \Eb\left[- \left\{\frac{\Delta \max(0, c(-Y+\tau)+1)}{\widehat S_C(Y|X,A)}  \right\} \frac{I\{A = d(X)\}}{\pi(A|X)}\right],\\
V^{\wedge,2}_\cL(f, c) & \equiv  \Eb \left[- \left\{\frac{\Delta \max(0, c(-Y+\tau)+1)}{\widehat S_C(Y|X,A)}  \right\} \frac{\cL(Af(\Xb))}{\pi(A|X)}\right],\\
V'^{,2}(d, c) & \equiv  \Eb\left[- \left\{\frac{\Delta \max(0, c(-Y+\tau)+1)}{\widetilde S_C(Y|X,A)}  \right\} \frac{I\{A = d(X)\}}{\pi(A|X)}\right],\\
V'^{,2}_\cL(f, c) & \equiv  \Eb \left[- \left\{\frac{\Delta \max(0, c(-Y+\tau)+1)}{\widetilde S_C(Y|X,A)}  \right\} \frac{\cL(Af(\Xb))}{\pi(A|X)}\right],
\end{align*}
and $(\widetilde f,\widetilde c)
\in \underset{f\in\cF,\,c\geq0}{\arg\max}
V'^{,2}_{\cL}(f,c)$.

We then have the following theorems for the buffered criterion. 

\begin{theorem}\label{2:fisher}
For $(\widetilde f,\widetilde c)$ defined above,
$(\sign(\widetilde f),\widetilde c)$ maximizes
$V'^{,2}(d,c)$ over all measurable treatment rules $d$ and $c\geq0$.
\end{theorem}

\begin{theorem}\label{2:excess}
For any $f\in\cF$ and $c\geq0$,
\begin{align*}
&\sup_{f'\in\cF,\,c'\geq0}
V'^{,2}(\sign(f'),c')
-
V'^{,2}(\sign(f),c)\leq
2\left\{
V'^{,2}_{\cL}(\widetilde f,\widetilde c)
-
V'^{,2}_{\cL}(f,c)
\right\}.
\end{align*}
\end{theorem}

\begin{theorem}\label{2:limit}
Suppose that \textbf{(A2.1)}--\textbf{(A2.3)} hold,
the kernel is bounded, and $\calH$ is dense in $L^1(P_X)$.
Assume that \(\inf_{x\in\cX,\,a\in\{-1,1\}}\pi(a|x)>0,\)
and let $\lambda=\lambda_n>0$ satisfy
$\lambda_n\to0$ and $n\lambda_n\to\infty$.
Suppose also that
\(
\sup_{x\in\cX,\,0\leq t\leq h}
|\widehat S_C(t|x,a)-S_C(t|x,a)|
\xrightarrow{p}0
\)
for both $a=-1,1$.
Let $(\widehat f,\widehat c)$ be a global solution
of \eqref{eq:ev2}, with $\lambda=\lambda_n$, and
assume that $\widehat c=O_p(1)$. Then
\[
\sup_{c\geq0}V^2(\sign(\widehat f),c)
\xrightarrow{p}
\sup_{f\in\calH,\,c\geq0}V^2(\sign(f),c).
\]
\end{theorem}

The proofs for all theorems are deferred to the supplementary material.

\section{Simulations}\label{sec:simu}

In this section, we conduct simulation studies under three scenarios. In each of them, we independently generate
$X^{(1)}, X^{(3)} \sim N(0,1)$ and
$X^{(2)} \sim \mathrm{Bernoulli}(0.5)$.
Treatment $A$ is randomly assigned with propensity $\Pr(A=1 \mid \Xb)=0.5$. Given $(X,A)$, we generate $\widetilde T$, set
$T=\min(\widetilde T,h)$ with $h=20$, generate the censoring time $C$,
and have $(X,A,Y,\Delta)$ with $Y=\min(T,C)$ and $\Delta=I(T\le C)$.

\textit{Scenario 1.}
Let $\varepsilon\sim\mathcal N(0,1)$ and $\widetilde T$ follows an accelerated failure time model
\[
\log \widetilde T
=
- X^{(1)} + 0.2 X^{(2)} + 0.8 X^{(3)}
+
\bigl(X^{(1)} + 0.5 X^{(2)} - 1.5 X^{(3)}\bigr) I\{A=1\}
+
\varepsilon.
\]
Censoring follows a Cox proportional hazards model with baseline hazard
$\lambda_0(t)=\tfrac12 t^{-1/2}$ and
\[
\lambda_C(t\mid X,A)=\lambda_0(t)\exp\{ -1 -0.8X^{(1)}-0.8X^{(2)}+0.4X^{(3)}
+\bigl(0.6-0.5X^{(1)}+0.3X^{(2)}-0.5X^{(3)}\bigr)I\{A=1\}\},
\]
yielding approximately $45\%$ censoring.

\textit{Scenario 2.}
Let  $\varepsilon\sim \mathrm{Weibull}(0.5, 0.3)$ and $\widetilde T$ follows an accelerated failure time model
\[
\log \widetilde T
=
-1.2 + 2.4X^{(1)} + 1.8X^{(3)}
+
\bigl(1.2 - X^{(2)} + X^{(1)} - 0.6X^{(3)}\bigr) I\{A=1\}
+
\varepsilon.
\]
Censoring again follows a Cox model with the same baseline hazard
$\lambda_0(t)=\tfrac12 t^{-1/2}$ and
\[
\lambda_C(t\mid X,A)=\lambda_0(t)\exp\{-1.5 + X^{(1)} + \bigl(-0.5+1.8X^{(1)}-0.6X^{(2)}\bigr)I\{A=1\}\},
\]
yielding approximately $35\%$ censoring.

\textit{Scenario 3.}
Let $\varepsilon\sim\mathcal N(0,1)$ and consider a Cox proportional hazards model for generating $\widetilde T$ with baseline
$\lambda_0(t)=\tfrac12 t^{-1/2}$ and 
\(
\lambda_{\widetilde T}(t\mid A,X,\varepsilon)=\lambda_0(t)\exp\bigl(\eta_{\widetilde T}(A,X,\varepsilon)\bigr),
\)
where
\[
\eta_{\widetilde T}(A,X,\varepsilon)=0.25+0.25X^{(1)}-0.15X^{(3)}+0.10\{X^{(2)}-0.5\}
+ I\{A=1\}\{-0.15+0.8\,\ell(X)\}+0.35\,\varepsilon,
\]
and $\ell(X)=0.9X^{(1)}+0.2\{X^{(2)}-0.5\}-0.9(X^{(3)}+0.5)$.
Censoring again follows a Cox model with the same baseline hazard
$\lambda_0(t)=\tfrac12 t^{-1/2}$ and
\[
\lambda_C(t\mid A,X)=\lambda_0(t)\exp\bigl(-0.2-1.5X^{(1)}+0.5X^{(2)}+\bigl(1.2-0.6X^{(1)}-1.4X^{(2)}-0.2X^{(3)}\bigr)I\{A=1\}\bigr),
\]
yielding approximately $50\%$ censoring.

Under each scenario, we generate $100$ independent training data sets
$\{(Y_i,\Delta_i,A_i,X_i)\}_{i=1}^n$ with $n=1000$.
For each training set, we apply Algorithm~1 to learn optimal treatment rules under the CVaR and buffered criteria, respectively. Throughout, we estimate the censoring survival function $S_C(\cdot\mid X,A)$ using survival forests \citep{athey2019generalized,cui2020estimating}, and we restrict attention to linear decision rules.
For the choice of $\gamma$ and $\tau$, we consider $\gamma= 0.25, 0.5, 0.5$ and
$\tau=0.5, 0.6, 0.2$ in Scenarios 1-3, respectively.

To evaluate a learned rule $d$, we generate an independent test set
$\{X_j\}_{j=1}^{N}$ with $N=10000$ and find the corresponding
potential survival times $T_j(d)$ under the known data generating mechanism.
We then compute empirical analogues of the three target value functions $V(d)$, $V^1(d)$, and $V^2(d)$.
For example, for $V^1(d)$, we report
\(
\frac{
\sum_{j=1}^N T_j(d) I\{T_j(d)\le \widehat Q_\gamma(T(d))\}
}{
\sum_{j=1}^N I\{T_j(d)\le \widehat Q_\gamma(T(d))\}
},
\)
where $\hat Q_{\gamma}\{T(d)\}$ is the empirical $\gamma$-quantile of
$\{T_j(d)\}_{j=1}^{N}$. The empirical value functions for $V(d)$ and $V^2(d)$ are defined
analogously.

We compare our proposed methods with two baselines aimed at finding mean-optimal treatment rules: causal survival forests (CSF) \citep{cui2020estimating} and CSF-based outcome-weighted learning (CSF-O) \citep{zhao2012estimating}, as well as approaches tailored to other decision criteria, including quantile (QuL) \citep{wang2017quantile} and expected quantile (EQuL) \citep{xiao2019robust}. Figure~\ref{sim1} summarizes the empirical values of $V(d)$, $V^1(d)$, and $V^2(d)$ across methods and scenarios. As expected, CSF and CSF-O perform competitively for the mean
criterion $V(d)$. In contrast, our CVaR-based procedure consistently achieves the largest $V^1(d)$, and our method for the buffered
criterion attains the best performance for $V^2(d)$.

\begin{figure}[ht]
  \centering
  
  \begin{subcaptionblock}{\textwidth}
    \centering
    \includegraphics[width=0.6\textwidth]{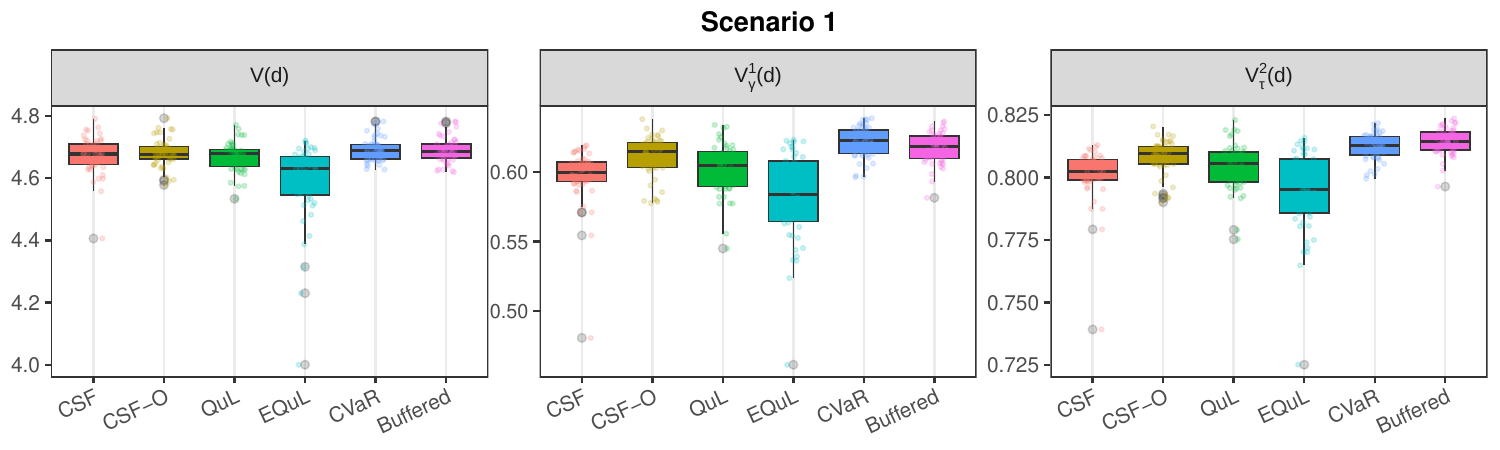}
  \end{subcaptionblock}
  
  \vspace{0.5em}
  
  \begin{subcaptionblock}{\textwidth}
    \centering
    \includegraphics[width=0.6\textwidth]{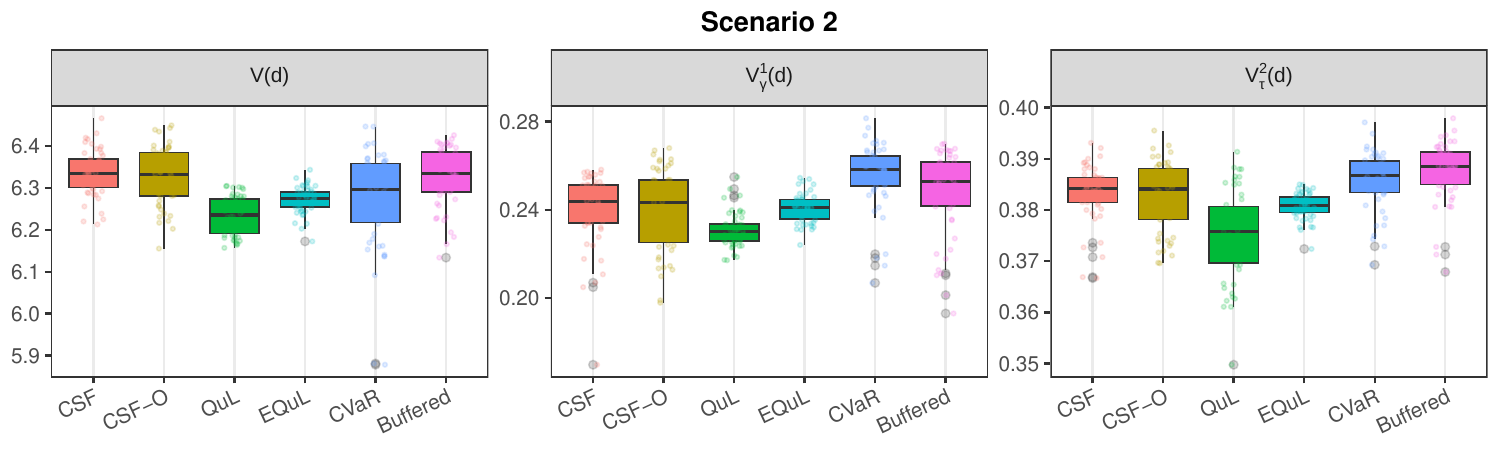}
  \end{subcaptionblock}
  
  \vspace{0.5em}
  
  \begin{subcaptionblock}{\textwidth}
    \centering
    \includegraphics[width=0.6\textwidth]{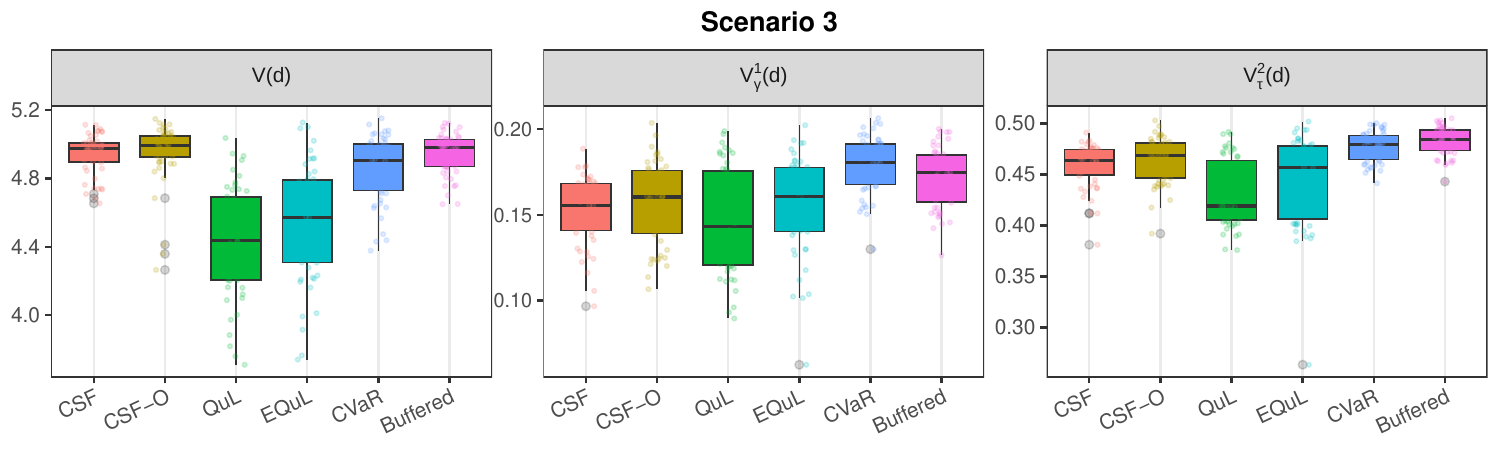}
  \end{subcaptionblock}

  \caption{Boxplots of empirical value functions.}
\label{sim1}
\end{figure}

Additional discussions of interpretability, sensitivity analyses, and further simulation results are provided in the supplementary material.

\section{Real data application}\label{sec:real}

We consider a randomized, double-blind, placebo-controlled trial AIDS Clinical Trials Group Protocol 175 (ACTG175) \citep{hammer1996trial} with survival time as a primary outcome. The original study included four treatment groups with patients randomly allocated: one group received 200 mg of zidovudine three times daily, one group was provided with 200 mg of zidovudine three times each day plus 0.75 mg of zalcitabine, one group was given 200 mg of zidovudine three times daily plus 200 mg of didanosine twice a day, while the remaining group was assigned 200 mg of didanosine twice daily.

Following \cite{cui2020estimating}, we focus on 1083 patients receiving treatment ZDV+ddI (
522 subjects) and ddI monotherapy (
561 subjects). The censoring percentage is around 78\%. The selected seven discrete baseline covariates include gender (904 male and 179 female), homosexual activity (725 yes and 358 no), race (772 white and 311 non-white), symptomatic status (192 symptomatic and 891 asymptomatic), history of intravenous drug use (142 yes and 941 no), hemophilia (92 yes and 991 no), and antiretroviral history (632 experienced and 451 naive). The selected five continuous baseline covariates are age, weight, Karnofsky score (scale of 0-100), CD4 count, and CD8 count.

Clinically, both vulnerable patients at risk of early failure and those with good long-term prognosis are of interest in HIV studies. Early mortality is strongly associated with uncontrolled viral replication, severe immunosuppression, and opportunistic infections, so treatment rules that reduce cases of very short survival are especially valuable. At the same time, modern antiretroviral therapy has transformed HIV into a chronic condition, and long-term survival and treatment durability have become key outcomes. It is therefore desirable to assess treatment rules not only by overall mean survival but also by how they behave in clinically important tails of the survival distribution, capturing both protection against early failure and support for sustained long-term survival. In addition, the empirical survival time distribution in ACTG175 exhibits a non-negligible fraction of early failures, which further motivates a tail-sensitive analysis.

Motivated by this, we compare methods that primarily aim to maximize mean survival, including CSF and CSF-O, with our methods targeting the CVaR and buffered criteria. Specifically, we set \(\gamma=0.50\) and \(\tau=460\). The choice $\gamma=0.50$ evaluates average survival over the lower half of the rule-induced survival distribution, whereas $\tau=460$ days, which is close to the first quartile of the observed event times, provides a clinically meaningful lower-tail reference level for early failure.

We use repeated $5$-fold cross-validation to evaluate each learned rule. In each replication, the data are randomly partitioned into five folds; four folds are used to train each method and the remaining fold is used for evaluation. This is repeated over all five test folds, and the entire procedure is in turn repeated $50$ times. For a generic treatment rule $d$, we denote the resulting cross-validated estimators by
\begin{align*}
\widetilde V(d)
&= \mathbb{P}_{\text{rep}}\left\{
\frac{
\mathbb{P}_{\text{test}}\left[
\frac{I\{A=d(X)\}}{\widehat{\pi}(A\mid X)}
\frac{\Delta Y}{\widehat S_C(Y\mid X,A)}
\right]
}{
\mathbb{P}_{\text{test}}\left[
\frac{I\{A=d(X)\}}{\widehat{\pi}(A\mid X)}
\frac{\Delta}{\widehat S_C(Y\mid X,A)}
\right]
}
\right\}, \\[6pt]
\widetilde V^1_{\gamma}(d)
&= \mathbb{P}_{\text{rep}}\left\{
\widehat{\alpha}_\gamma -
\frac{
\mathbb{P}_{\text{test}}\left[
\frac{I\{A=d(X)\}}{\widehat{\pi}(A\mid X)}
\frac{\Delta(\widehat{\alpha}_\gamma-Y)_+}
{\gamma\widehat S_C(Y\mid X,A)}
\right]
}{
\mathbb{P}_{\text{test}}\left[
\frac{I\{A=d(X)\}}{\widehat{\pi}(A\mid X)}
\frac{\Delta}{\widehat S_C(Y\mid X,A)}
\right]
}
\right\}, \\[6pt]
\widetilde M^2_{\tau}(d)
&= \mathbb{P}_{\text{rep}}\left\{
\frac{
\mathbb{P}_{\text{test}}\left[
\frac{I\{A=d(X)\}}{\widehat{\pi}(A\mid X)}
\frac{\Delta\max\{0,\widehat c_\tau(\tau-Y)+1\}}
{\widehat S_C(Y\mid X,A)}
\right]
}{
\mathbb{P}_{\text{test}}\left[
\frac{I\{A=d(X)\}}{\widehat{\pi}(A\mid X)}
\frac{\Delta}{\widehat S_C(Y\mid X,A)}
\right]
}
\right\},
\end{align*}
where $\mathbb{P}_{\text{test}}$ denotes the empirical mean over the test fold, $\mathbb{P}_{\text{rep}}$ denotes the empirical mean over all repetitions. For each split, the nuisance components $\widehat S_C(\cdot|X,A)$ and $\widehat \pi(\cdot|X)$, as well as
$(\widehat\alpha_\gamma,\widehat c_\tau)$, are fit using only the corresponding training folds, and the nuisance components are then evaluated on the test fold when computing
$\widetilde V(d)$, $\widetilde V^1_\gamma(d)$, and $\widetilde M^2_\tau(d)$. Since we only have offline data from the study, these $\widetilde V(d)$, $\widetilde V^1_\gamma(d)$, and $\widetilde M^2_\tau(d)$ provide attainable estimators of the corresponding value functions. In particular, $\widetilde V(d)$ is the usual estimator of mean survival under $d$, while $\widetilde V^1_\gamma(d)$ and $\widetilde M^2_\tau(d)$ approximate the results introduced in Section~\ref{subsec:rc}, evaluated at the chosen values of $\gamma$ and $\tau$. It is noteworthy that larger values of $\widetilde V(d)$ and $\widetilde V^1_{\gamma}(d)$ and smaller values of $\widetilde M^2_{\tau}(d)$ are preferred.

\begin{table}[!htbp]
\centering
\footnotesize

\begin{minipage}[t]{0.48\textwidth}
\centering
(a) Value summary

\vspace{2mm}
\begin{tabular}{lccc}
\toprule
Method
& $\widetilde V(d)$
& $\widetilde V^1_{0.50}(d)$
& $\widetilde M^2_{460}(d)$ \\
\midrule
CSF
& 624.57
& 394.85
& 0.677 \\

CSF-O
& \bfseries 635.14
& 412.18
& 0.617 \\

CVaR
& 635.12
& \bfseries 413.73
& 0.611 \\

Buffered
&  635.04
&  413.24
& \bfseries 0.608 \\
\bottomrule
\end{tabular}
\end{minipage}
\hfill
\begin{minipage}[t]{0.48\textwidth}
\centering
(b) Selected coefficient summary

\vspace{2mm}
\begin{tabular}{lrr}
\toprule
Covariate
& CVaR
& Buffered \\
\midrule
Prior antiretroviral experience
& -7.767
& -7.061 \\

Symptomatic status
& 7.995
& 7.243 \\

Weight
& -6.611
& -5.948 \\

CD4 count
& 0.128
& 0.330 \\
\bottomrule
\end{tabular}
\end{minipage}

\caption{Cross-validated value estimates and selected coefficients of the
estimated treatment rules for the ACTG175 application. Panel (a) reports
the repeated cross-validation results. Larger values of
$\widetilde V$ and $\widetilde V^1_{0.50}$ and smaller values of
$\widetilde M^2_{460}$ are preferred. Panel (b) reports coefficients
obtained by refitting the rules using the full analysis sample.}
\label{tab:hiv_main_results}
\end{table}

\st{The left panel in Table~\ref{tab:hiv_main_results} reports the cross-validated mean value $\widetilde V(d)$ together with the CVaR criterion ($\widetilde V^1_{0.50}(d)$) and the buffered criterion ($\widetilde M^2_{460}(d)$) for each rule. The CVaR criterion-oriented rules achieve the largest values of $\widetilde V^1_{0.50}(d)$, and the buffered criterion-oriented rules exhibit lower $\widetilde M^2_{460}(d)$ values compared to those of the benchmarks. To further interpret these results, we report selected coefficients of the estimated treatment rules in the right panel of Table~\ref{tab:hiv_main_results}.
While the two rules exhibit similar coefficient patterns, the buffered rule places relatively greater emphasis on baseline CD4 count, a measure of patients' immunological status, whereas the
CVaR rule assigns slightly greater relative weight to body weight. These patterns illustrate how distinct tail-sensitive objectives may shape treatment allocation. Complete coefficient estimates and further discussion are provided in the supplementary material.}

\section{Discussion}\label{sec:discussion}
In this paper, we propose a novel framework to estimate the optimal treatment rules in order to control the tail of the subjects' survival time with right-censored data. The proposed methods target directly the truncated mean survival time and survival probabilities, which are of great interest in survival analysis. For the first criterion, we maximize the truncated mean survival time by choosing a user-specified quantile endpoint; for the second criterion, we maximize survival function at a quality-adjusted survival time, where the time is defined by the first value function. Moreover, the proposed methods have intrinsic connections to CVaR and bPOE. Thus, the resulting optimal rules can potentially prevent adverse events.

The proposed methods admit several extensions. First, as discussed in the supplementary
material, a multi-arm trial can be accommodated without changing the identification arguments; for estimation, we
discuss an angle-based multicategory representation and a smooth logistic
relaxation, while noting that constructing a multicategory DC surrogate is substantially more challenging
\citep{qiao2009adaptive,JMLR:v18:17-003,Zhou2018OutcomeWeightedLF,qi2019multi,xue2022multi,zhou2022offline}.
Second, personalized dose finding with explicit tail control is a natural
direction \citep{chen2016dosefinding,zhou2021parsimonious,fang2023fairness}.
Third, extending our criteria to a multi-stage (e.g., mobile health) decision-making process will require longitudinal weighting to handle time-varying treatment
assignment and censoring, and will introduce additional non-smooth structure
that complicates scalable optimization \citep{zhao2015new,shi2018high,luckett2019,shi2022statistical,ye2024stage,zhou2024estimating,li2026reinforcement}.
Finally, it would be valuable to develop augmented inverse probability weighted,
doubly robust estimators for our proposed value functions \citep{robins1994estimation,zhang2012robust,zhao2015doubly,liu2018augmented}.

\section*{Acknowledgments}

The authors are thankful to the four referees, associate editor, and editor for helpful comments that led to an improved manuscript. This research was partially supported by the National Key R\&D Program of China (2024YFA1015600). The authors report no competing interests to declare.

\section*{Data availability} 
The data and code are available at \url{https://github.com/taoshen2022/Robust-Survival-ITR}.

\linespread{0.9}\selectfont
\bibliographystyle{asa}
\bibliography{survitr}

\end{document}